\documentclass[conference]{IEEEtran}
\IEEEoverridecommandlockouts
\usepackage{cite}
\usepackage{amsmath,amssymb,amsfonts}
\usepackage{graphicx}
\usepackage{textcomp}
\usepackage{xcolor}

\usepackage{amsmath, amssymb, amsthm,algorithm}
\usepackage[latin1]{inputenc}

\usepackage{latexsym}
\usepackage{graphics,epsfig}
\usepackage{amsfonts}
\usepackage{booktabs}
\usepackage{color}
\usepackage{multirow}
\usepackage[justification=centering]{caption}
\usepackage{graphicx}
\usepackage{adjustbox}
\usepackage{kantlipsum}
\usepackage[caption=false,font=footnotesize]{subfig}
\usepackage{cases}
\usepackage{url}
\usepackage{tabu}
\usepackage{makecell}
\usepackage{soul}

\usepackage{array}
\usepackage{algpseudocode}
\usepackage{mathtools,lipsum,cuted}

\usepackage[justification=centering]{caption}

\usepackage[caption=false,font=footnotesize]{subfig}
\usepackage[T1]{fontenc}
\usepackage{mwe}
\usepackage{subfig}

\newlength{\tempheight}
\newlength{\tempwidth}

\newcommand{\rowname}[1]
{\rotatebox{90}{\makebox[\tempheight][c]{#1}}}

\newcommand{\columnname}[1]
{\makebox[\tempwidth][c]{#1}}

\makeatletter
\def\BState{\State\hskip-\ALG@thistlm}
\makeatother

\newcounter{example}[section]

\usepackage[english]{babel}
\usepackage{verbatim}

\usepackage[english]{babel}

\theoremstyle{plain}
 \newtheorem{thm}{Theorem}

 \newtheorem{lemma}{Lemma}

 \newtheorem{corollary}{Corollary}

\theoremstyle{remark}

\def\BibTeX{{\rm B\kern-.05em{\sc i\kern-.025em b}\kern-.08em
    T\kern-.1667em\lower.7ex\hbox{E}\kern-.125emX}}
\begin{document}

\title{Benchmarking Buffer Size in IoT Devices Deploying REST HTTP\\
}

\author{\IEEEauthorblockN{Cao Vien Phung, Mounir Bensalem and Admela Jukan}
\IEEEauthorblockA{Technische Universit\"at Braunschweig, Germany \\
\{c.phung, mounir.bensalem, a.jukan\}@tu-bs.de}

}

\maketitle

\begin{abstract}
A few potential IoT communication protocols at the application layer have been proposed, including MQTT, CoAP and REST HTTP, with the latter being the protocol of choice for software developers due to its compatibility with the existing systems. We present a theoretical model of the expected buffer size on the REST HTTP client buffer in IoT devices under lossy wireless conditions, and validate the study experimentally. The results show that increasing the buffer size in IoT devices does not always improve  performance in lossy environments, hence demonstrating the importance of benchmarking the buffer size in IoT systems deploying REST HTTP.
\end{abstract}

\section{Introduction}
IoT (Internet of Things) systems use a few communication protocols to connect sensors and actuators with the IoT data and routing hubs, including MQTT (Message Queue Telemetry Transport) \cite{OASIS2014}, CoAP (Constrained Application Protocol) \cite{Shelby2014} and REST HTTP  (Representational State Transfer and Hypertext Transfer Protocol) \cite{Severance2015}. The latter includes two layers, the upper layer being REST and the lower layer HTTP. HTTP is the fundamental client-server model protocol, which is widely used by developers today due to its compatibility with existing network infrastructure \cite{10.1145/3292674}. The client-server model of HTTP protocol is performed by sending a request message to the server by the client, and then returning a corresponding acknowledgement back to the client by the server if that request message was accepted. REST is based on a specific architectural style with the guideline for web services developments to define the interaction among different components, and has been recently combined with HTTP protocol and deployed in IoT-based systems  \cite{10.1145/3292674}.

When using REST HTTP in IoT systems, the analysis of buffer size is critical because it directly impacts the communication performance, whereby larger size buffers are expected to reduce data packet blocking and thus increase transmission reliability. On the other hand, larger buffers can constrain the hardware architecture and physical size of the IoT devices. Selecting an inappropriate buffer length. Hence, there a tradeoff between communication performance and hardware design, requiring proper benchmarking of the buffer size.  Previous work, such as \cite{Heidemann:1997:MPH:268715.268719} focused on some related important aspects such as interaction with underlying transport protocols, HTTP performance over satellite channels  \cite{HansKruse}, and presentation of an approximate analytical model related to underlying transport layer \cite{1258974}. Other works, such as \cite{8400067} analyzed the impact of pipelining on the HTTP latency, while \cite{8756782} and \cite{9149026} analyzed the amount of redundant REST HTTP data in environments with unstable connections, and yet \cite{7226719} focused on the impact of HTTP pipelining. Of special interest is related work \cite{inproceedings} which analyzed the receiving buffer size of IoT edge router (server), based on the analysis of the average TCP (Transmission Control Protocol) congestion window size of all IoT devices. This paper focused on scenarios where data streams arrive from a large number of IoT devices and experience packet losses due to congestion events. 

This paper benchmarks the buffer size in IoT devices deploying REST HTTP, and focuses on the client side. We develop a novel analysis, under the assumption of IoT device availability churns, resulting in intermittent and unreliable communication. These events can cause network volatility with unusually high latency and fully closed connections, therefore retransmissions for timeout events set for RESTful applications are necessary with the aim of fault tolerance \cite{6296043}. Our Markov chain model analyzes this and derives the impact of message losses, caused by unusually high delay from node churns on expected buffer size on the client buffer of IoT devices.  The analytical results are validated by the experiments in a simple testbed with IoT devices experiencing node churn. The results show increasing the buffer size in IoT devices does not always improve  performance in lossy environments, in other words, that does not always decrease the amount of arrival data blocked at the client side, hence demonstrating the importance of proper buffer size benchmarking in IoT systems.
\section{Analytical Model}\label{Analysis}
We assume that IoT client devices experience the so-called node churn,  meaning that a device can appear and disappear from the system, due to changed location, depletion of the battery or loss of wireless connectivity. Node churn can cause variable latency, leading to generating automatically retry request messages due to timeout events set for RESTful applications. In the example shown in Fig. \ref{scenario}, one REST HTTP client, e.g., mobile phone, tries to add resources with POST methods on one REST HTTP server, e.g., edge node, using stop-and-wait mechanism. We assume that the arrival time interval between two POST requests is constant, e.g., each $i^{th}$ arrival request message $q_i$ updated at a time at the client buffer with a constant arrival time interval is $t=1$s, where $q_i$ is generated by the client REST layer using JSON file format, which is often used in IoT standards over HTTP. Before transmitting the JSON file of $q_i$, HTTP layer adds a header to POST method, which contains the HTTP version, the Content-type, and the root of the resource, etc. Request messages might have different sizes depending on the programming technology of the web server \cite{8088251}. We assume that the data formats from the same application are often similar \cite{7746084}, where request messages have the same size. We finally define the \emph{observations} as request message read processes right after either timeout event or arrival of the first request message event when the client buffer is empty; in other words, assuming there always exists at least one request message in the client buffer. Timeout time interval for each RESTful application in this example is set to $T_o=2$s.  Arrival messages are stored at the client REST buffer and only deleted when they are successfully sent and their  corresponding responses are successfully returned, whereby \cite{6296043} gave an example of retransmissions with $5$ times for each message failed; however, we assume unlimited retransmissions. 

 In this example, the client buffer length is fixed to $M=km+1$, defined to be the maximum number of messages that the client can buffer, whereby $k \in \{1,2\}$ is considered and $m=\frac{T_o}{t}\in \mathbb{N}^*$ is the number of arrival request messages updated at the client side after a timeout event, e.g., in Fig. \ref{scenario}, $k=2$ and $m=\frac{T_o}{t}=2$, i.e., $M=5$. We assume that new messages arrived at the client side are ignored when previous messages are stored in the client buffer, and the client can send  successfully  the stored  messages with smooth connections. This can be explained by the fact that the time interval  $t$ between two request messages arrived at the client  is much higher than the transmission time of stored messages. Theoretical and experimental results presented in the next section confirm our assumption.
\begin{figure}[!ht]
\centering
\includegraphics[width=1 \columnwidth]{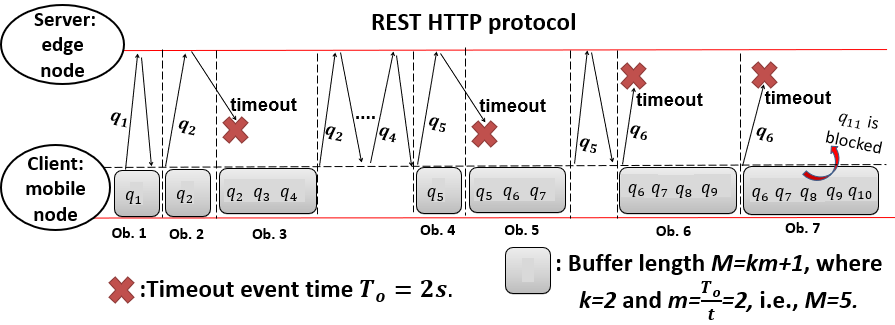}%
  \vspace{-0.2cm}
  \caption{Scenario of REST HTTP  in lossy environment, where Ob. $j$ is observation $j$, $q_i$ denotes request message $i$ and $t=1$s represents arrival time interval of each request message updated at a time at the client side.}
  \vspace{-0.2cm}
  \label{scenario}
\end{figure}

In Fig. \ref{scenario}, request message $q_1$ arriving at Ob. $1$ (Ob. $j$ denotes Observation $j$ for short) is removed from the  buffer when its response is successfully returned on the first try. For $q_2$ arriving at Ob. $2$, assuming its response is lost on the way back to client, the client keeps buffering  this message along with $m=2$ arrival messages of $q_3$ and $q_4$, after timeout event of Ob. $3$. Assume that any timeout event occurs due to network volatility causing fully closed connections. When $q_2$, $q_3$ and $q_4$ are removed, $q_{5}$ stored in the buffer arrives at its arrival time of Ob. $4$, but similarly, assuming its response is lost on the way back to client, the client keeps buffering  it along with $m=2$ arrival request messages of $q_6$ and $q_7$, after timeout event of Ob. $5$. $q_5$ is only removed from the buffer when its response is successfully returned in the second attempt. When  $q_6$ is lost with its first timeout event, $4$ request messages are buffered at Ob. $6$, i.e., from $q_6$ to $q_9$. As $M=5$, at the second timeout event of $q_6$, i.e., Ob. $7$, only $q_{10}$ is added into the buffer and $q_{11}$ is blocked, i.e., $q_6$ to $q_{10}$ are buffered. 
\subsection{State transition probability of Markov chain model} \label{MCM} 
In this subsection, we analyze the state transition probability of Markov chain model, presented in Fig. \ref{markovchain}.
\subsubsection{$k=2$}
     \begin{table}[t!] \vspace{0.2cm}
  \centering
  \caption{List of notations.}
  \label{tab:table1}
  \begin{tabular}{ll}
    \toprule
    Notation & Meaning\\
    \midrule
    $p$ & Failed message probability.\\
    $m$ & Arrival messages updated at the client in a timeout event.\\
     $\pi(e)$ & Steady-state probability of $e$ requests.\\
      $E$ & State space.\\
	 $\delta_{u \rightarrow  v}$ & Transition probability from state $u$ to state $v$.\\
	 $M$ & Client buffer length. \\
	 $S$ & Expected client buffer size. \\
	 $T_o$ & Time interval for a timeout event.\\
	 $t$ & Arrival time of each message updated at a time at the client.\\
    \bottomrule
  \end{tabular}\vspace{-0.5cm}
\end{table}
Fig. \ref{markovchain} shows the Markov chain model with $k=2$ reflecting the above defined REST HTTP scenario, whereby $p$ denotes the failed request message probability. The state space with $k=2$ is $E=\{1,m+1, m+2,...,2m+1\}$, including $m+2$ states, where a state is defined to be the number of requests in the client buffer. The notations are summarized in Table \ref{tab:table1}. The transition probability from state $u=1$ to state $v=1$ and $v=m+1$ is:
 \begin{subnumcases}{\delta_{1\rightarrow v}=}
    1-p &, $v=1$ \label{1tov1}\\
p &, $v=m+1$   
\label{1tov2}
\end{subnumcases}
 Eq. \eqref{1tov1} shows state $1$ keeps staying itself when only one message is stored in the client buffer and it is successfully sent with the transition probability $\delta_{1 \rightarrow 1}=1-p$. Since we observe arrival of the first request message event when the client buffer is empty, there is no state of $0$, e.g., in Fig. \ref{scenario} after $q_1$ arriving at Ob. $1$ is successfully sent, only $q_2$ stored in the client buffer arrives at its arrival time of Ob. $2$, i.e.,  state $1$ remains after Ob. $1$ and Ob. $2$ with transition probability $\delta_{1 \rightarrow 1}=1-p$. 
 
 Eq. \eqref{1tov2} is state $u=1$ to directly transits to state $v=m+1$ when only one message is buffered and its timeout event occurs with $\delta_{1 \rightarrow m+1}=p$. In Fig. \ref{markovchain} there are no direct state transitions from state $1$ to state $v \in [m+2,2m+1]$ and that states $v \in (1,m]$ do not exist; the reason is that the buffer increases with the maximum number of $m$ arrival messages at the client after each timeout event, while at least one message is always buffered. For instance, in Fig. \ref{scenario} after timeout event of Ob. $3$, of $3$  messages in the buffer, i.e.,  state $u=1$ at Ob. $2$ directly transits to state $v=m+1=3$ at Ob. $3$ with $\delta_{1 \rightarrow 3}=p$, where $m=2<v=m+1=3<m+2=4$; in other words, there is no state $m=2$ and there are no direct state transitions from state $u=1$ to state $v \in \{4;5\}$.
 
The transition probability from state $u \in [m+1 , 2m+1]$ to state $v=1$ and $v \in [m+1,2m+1]$ is given by
\begin{subnumcases}{\delta_{u\rightarrow v}=}
(1-p)^u &, $v=1$ \label{lasttov1}\\
(1-p)^{u-v+m}  p &, $v \in [m+1,2m]$  \label{lasttov2}\\
\sum_{i=0}^{u+m-v} (1-p)^i p  &, $v=2m+1$
\label{lasttov3}
\end{subnumcases}

 Eq. \eqref{lasttov1} denotes state $u\in [m+1,2m+1]$ directly transits to state $1$ when all messages in the client buffer are successfully sent with $\delta_{u \rightarrow 1}=(1-p)^u$, e.g., $\delta_{m+1 \rightarrow 1}=(1-p)^{m+1}$. Consider the example in Fig. \ref{scenario} again, when $q_2$, $q_3$ and $q_4$ buffering at Ob. $3$ are removed, while $q_5$ stored in the buffer arrives at its arrival time of Ob. $4$, i.e., state $u=m+1=3$ at Ob. $3$ transits to the state $v=1$ at Ob. $4$ with $\delta_{3 \rightarrow 1}=(1-p)^3$.
 
Eq. \eqref{lasttov2} indicates that state $u\in [m+1,2m+1]$ can directly transit to state $v \in [m+1,2m]$, where the oldest $u-v+m$ consecutive request messages are successfully sent with probability $(1-p)^{u-v+m}$, but $(u-v+m+1)^{th}$ message is failed with loss probability $p$. Therefore, the transition probability from state $u$ to state $v$ in this case is $\delta_{u \rightarrow v}=(1-p)^{u-v+m}  p$, e.g., in Fig. \ref{markovchain}, $\delta_{2m \rightarrow m+2}=(1-p)^{2m-2}  p$. As previously discussed, there are no states $v \in (1,m]$; and there are also no state transitions from $u\in [m+1,2m+1]$ to the states $v \in (1,m]$. For example in Fig. \ref{scenario}, one message $q_5$ belonging to Ob. $5$ is successfully removed at the client, and under the timeout of $q_6$ at Ob. $6$, the client buffers $4$ messages, i.e., state $u=m+1=3$ at Ob. $5$ directly transits to state $v=m+2=2m=4$ ($4>m=2$) at Ob. $6$ with $\delta_{3 \rightarrow 4}=(1-p)  p$.

Eq. \eqref{lasttov3} shows that state $u \in [m+1,2m+1]$ directly transits to state $v=2m+1$, whereby transition probability becomes $\delta_{u \rightarrow 2m+1}=\sum_{i=0}^{u-m-1}(1-p)^ip$. Since the buffer length is limited to $M=2m+1$, any state $u$ under this consideration cannot transit to states larger than $2m+1$, which means that a few messages are to be blocked at the client, if the total number of messages need to be stored in more than $2m+1$ of buffer length. For instance, in Fig. \ref{markovchain}, transition probability from state $u=m+2$ to state $v=2m+1$ is $\delta_{m+2 \rightarrow 2m+1}=p+(1-p)p$, i.e., if the oldest message fails or succeeds but the second oldest one fails; then, the state of buffer $u=m+2$ always transits to state $v=2m+1$ because the buffer increases with the maximum number of $m$  arrival request messages after each timeout event, but is always limited to $M=2m+1$. In this example, one arrival request message is blocked at the client after timeout event in case of the oldest request message unsuccessfully sent, because after the timeout event the buffer needs to store $2m+2$ messages, while it is limited to $M=2m+1$, e.g., in Fig. \ref{scenario}, at the timeout event of Ob. $7$, as $M=2m+1=5<6$, $q_{11}$ is blocked, i.e., state $m+2=4$ at Ob. $6$ directly transits to state $2m+1=5$ at Ob. $7$; while no message is blocked in case of the oldest message successfully sent but the second oldest one failed because the client buffer removes the oldest one after its success, i.e., in this case, the client buffer length is sufficient to storing $2m+1$ messages. Note that in case, transition probability from state $u=m+1$ to  state $v=2m+1$ can be also applied by Eq. \eqref{lasttov2}, i.e.,  $\delta_{m+1 \rightarrow 2m+1}=p$.
 \begin{figure}[!ht]
\centering
\includegraphics[width=1\columnwidth]{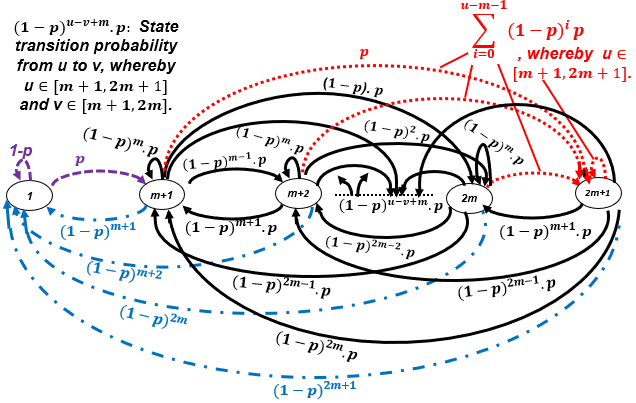}%
  \vspace{-0.05cm}
  \caption{Markov chain model with $k=2$ for the client buffer with REST HTTP in lossy IoT system.}
  \vspace{-0.4cm}
  \label{markovchain}
\end{figure}
\subsubsection{$k=1$} \label{specialcasek1}
With $k=1$, the state space only includes two states to be $E=\{1,m+1\}$. The transition probability from $1$ to itself and $m+1$ are Eq. \eqref{1tov1} and Eq. \eqref{1tov2}, respectively. The transition probability from the state $m+1$ to $1$ and itself are Eq. \eqref{lasttov1} and Eq. \eqref{lasttov3}, respectively. 

\subsection{Lemmas and theorems based on the Markov chain model}
\begin{lemma} \label{fia}
The Markov chain model with $k=1$ and $k=2$  is finite, irreducible and aperiodic.
\end{lemma}

\begin{proof}
The state space of these Markov chains is finite as it has $2$ states for $k=1$ and $m+2$ states for $k=2$. Additionally, since all states communicate together: For $k=1$, the two states of $1$ and $m+1$ can be reached together with a non-zero probability after $1$ transition step; while for $k=2$ in Fig. \ref{markovchain}, a state $u \in [m+1,2m+1 ]$ and $u=1$ are reachable from any other state $v \in \{1,m+1,m+2,...,2m+1 \} \setminus \{u\}$ and $v=m+1$, respectively, with a non-zero probability after $1$ transition step, and $u=1$ is reachable from any other state $v \in [m+2,2m+1 ]$ with a non-zero probability after $2$ transition steps; these Markov chains are irreducible. Finally, we prove that these Markov chains are aperiodic by proving that any state $u$ is aperiodic. The period $d(u)$ of a state $u$ is the greatest common denominator of all integers $n>0$, for which the $n-$step transition probability from state $u$ to itself $P_{uu}^{(n)}>0$. As each state $u$ of these Markov chain models always has a self-transition  with a non-zero probability, $d(u)=1$ and consequently state $u$ is aperiodic.
\end{proof}
As the Markov chains of REST HTTP in lossy environments with $k\in\{1,2\}$ satisfy Lemma \ref{fia}, they have a steady-state distribution to which the distribution converges, starting from any initial state. The total probability of all steady-states is $\sum_{e \in E} \pi(e)=1$, whereby $\pi(e)$ is the steady-state probability of $e$ messages in the buffer. With section \ref{MCM}, we build the balance equations of any state defined that the total probability of entering it is equal to the total probability of leaving it.

\textbf{Case of $k \geq 1$ for the state $1$,} based on the analysis in section \ref{MCM} and using Eq. \eqref{1tov2} and Eq. \eqref{lasttov1},  we have the balance equation Eq. \eqref{pt1} of state $1$ for use case $k \geq 1$:
 \begin{equation}\label{pt1}
\begin{split}
\pi(1) \cdot p= \sum_{i=m+1}^{km + 1} \pi(i) \cdot (1-p)^{i}
 \end{split}
 \end{equation}
\subsubsection{For $k=1$}
 Eq. \eqref{pt1} with $k=1$ is the balance equation of state $1$. Based on the analysis in section \ref{MCM} and using Eq. \eqref{1tov2} and Eq. \eqref{lasttov1}, we have the balance equation Eq. \eqref{pt2} of state $m+1$ for $k=1$: \begin{equation}\label{pt2}
\begin{split}
\pi(m+1) \cdot (1-p)^{m+1}= \pi(1) \cdot p 
 \end{split}
 \end{equation}
\begin{lemma} \label{kequal1}
For $k=1$, the probability of  steady-states is:
\begin{equation}
\pi(1) = \frac{(1-p)^{m+1}}{p+(1-p)^{m+1}} \; ; \; \pi(m+1) = \frac{p}{p+(1-p)^{m+1}}
 \end{equation}
\end{lemma} 
 \begin{proof}
Using Eq. (\ref{pt2}) and $\sum_{e \in E} \pi(e)=1$ to solve.
 \end{proof}
 \begin{thm} \label{theorem:ABS}
 The expected client buffer size $S$, for k=1,  is:
 \begin{equation}
 S = \frac{(1-p)^{m+1} + (m+1)p}{p+(1-p)^{m+1}}
 \end{equation}
 \end{thm}
 \begin{proof}
Using lemma \ref{kequal1} and $S = \sum_{e \in E} e\pi(e)$ to solve.
 \end{proof}
\subsubsection{For $k=2$} Eq. \eqref{pt1} with $k=2$ is the balance equation of state $1$. The balance equation of state $m+1$ for $k=2$ is referred to Appendix \ref{appen1}, and it is simplified by  using Eq. (\ref{pt1}):
\begin{equation}\label{pt2-4}
\pi(m+1) =   \pi(1)\cdot p  (1-p)^{-1}
 \end{equation} 

For all $e \in [m+2,2m]$, the balance equations of state $e$ for the case of $k=2$ can be referred to Appendix \ref{appen2} and Appendix \ref{appen3}, and then they are simplified by using Eq. \eqref{pt1} with $k=2$: 
\begin{equation}\label{pt3-2}
\pi(e)=\pi(1)\cdot  p^{2}(1-p)^{m-e}
 \end{equation}

The balance equation of state $2m+1$ for the case of $k=2$ can be referred to Appendix \ref{appen4}, and then it is simplified by applying Eq. \eqref{pt1} with $k=2$ and $\sum_{e \in E} \pi(e)=1$:
\begin{equation}\label{pt3-6}
\begin{split}
&\pi(2m+1) = 1 - [1+ p(1-p)^{-m}]\pi(1)
 \end{split}
 \end{equation}
 \begin{lemma}\label{lemma:21k2}
Let $r=1-p$ and $\breve{D}=p-r^m[p+(m-1)p^2-r^{m+1}-pr]$, for $k=2$, the probability of  steady-states is:
\begin{equation}\label{eq:k2pi1}
\pi(1) = \frac{r^{2m+1}}{\breve{D}} \; ; \; \pi(m+1) = \frac{pr^{2m}}{\breve{D}} \; ; \; \pi(e)=  \frac{p^2r^{3m-e+1}}{\breve{D}} \; ; \; \notag
\end{equation}
\begin{equation}\label{eq:k2piv}
\begin{split}
\pi(2m+1)= \frac{p[1-r^m(1+(m-1)p)]}{\breve{D}} \; ; \; \forall m+2\leq e\leq 2m \notag
\end{split}
\end{equation}
\end{lemma}
 \begin{proof}
Using Eqs. (\ref{pt1}), (\ref{pt2-4}),  (\ref{pt3-2}) and (\ref{pt3-6}) to solve.
 \end{proof}
 \begin{thm} \label{ABS:theoremk2}
  $r=1-p$, the  expected buffer size $S$ for k=2  is:
\begin{equation}\label{eq:BSk2}
S = \frac{ r^m[r(2r^m-1)-mp^2(2m+1)]+(2m+1)p }{p-r^m[p+(m-1)p^2-r^{m+1}-pr]},
\end{equation}
 \end{thm}
 \begin{proof}
Using $S = \sum_{e \in E} e\pi(e)$ and Lemma \ref{lemma:21k2} to solve.
 \end{proof}
  \begin{corollary} \label{lexuan}
For all $k$, the limit of expected buffer size $S$  is:
 \begin{equation}
\lim_{p \rightarrow 0} S = 1  \; \;  ; \; \; \lim_{p \rightarrow 1} S = km+1
 \end{equation}
\end{corollary}
 \begin{proof}
We take Theorem \ref{theorem:ABS} and \ref{ABS:theoremk2} to solve, where $p \rightarrow 0;1$.
 \end{proof}
\subsubsection{For $k\geq 3$} With  different values of $k$, we have different functions of expected client buffer size $S$ and probability of steady-states $\pi(e)$. This requires complex and major repeating calculations for each different value of $k$.
\section{Validation and Benchmarking}\label{results}
\subsection{Experimental setup}
The experimental validation setup consists of a simple wired network  for communication between a  client and a server with REST HTTP (Fig. \ref{setupnetwork}). It includes a laptop (Intel(R) Core(TM) i7-3667U CPU @ 2.00GHz) representing our mobile client, a Raspberry Pi 3 Model B+ (64-bit quad-core ARM Cortex-A53 processor @ 1.4GHz and 1 GB of RAM) representing the static server running python flask version 1.0.3, and an access point router used to connect the laptop with the Raspberry Pi. We consider a wired network to obtain a stable implementation and to simplify the setup.
\begin{figure}[ht]
  \centering
  \subfloat[An experimental wired network setup.]{\includegraphics[ width=5cm, height=1.9cm]{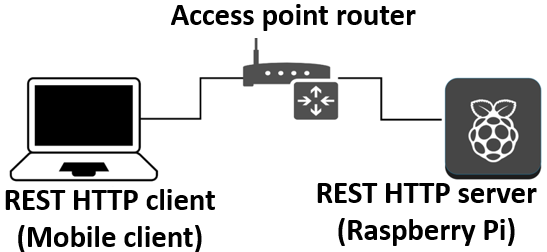}
  \label{setupnetwork}}
  \subfloat[JSON file \cite{9149026}]{\includegraphics[ width=3.5cm, height=1.9cm]{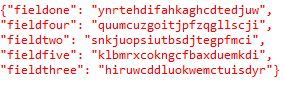}
  \label{json}}
  \vspace{-0.1cm}
    \caption{Experimental setup.}
  \label{experiments}
  \end{figure}
In environments with unreliable connections, we consider stop-and-wait mechanism and the number of retransmissions  unlimited. In order to emulate the connection losses in our wireless scenario, we incorporate  the wireless losses  into the experimental setup using  Algorithm \ref{losscon} \cite{9149026}, where $p$ denotes the message loss probability. This emulation is implemented as follows: We randomly generate a vector of binary numbers, where $20\%$ of the values are ones, e.g., $vector=(0,0,0,1,0)$, where $1$ represents the connection loss until expired timeout and $0$ represents no connection losses; the proposed assumption means that $20\%$ of links are lost. Using the link loss vector  for Algorithm \ref{losscon}, if the client or server meets the value of $1$, then the transmission is paused for a period of time $t^{,}$ to represent a connection loss, e.g., $t^{,}$ is set so that it is equal to the value of $timeout$ $T_o$ set for the client or $t^{,}$ is long enough so that the response message cannot be returned to the client side. Algorithm \ref{losscon} is required to be executed  before dispatching the request and response message. When timeout expires, REST HTTP client forgets the previous connection and opens a new one for retransmission.
\begin{algorithm}\caption{Emulation of link loss until timeout expires \cite{9149026}}
\begin{algorithmic}[1]
\Procedure{Solve}{$Link\_loss\_vector$} 
\If{we meet $1$ from $Link\_loss\_vector$}
\State Pause ($t^{,}$), e.g., $t^{,}=timeout$
\EndIf
\EndProcedure
\end{algorithmic}
\label{losscon}
\end{algorithm}
 \begin{figure*}[ht]
  \centering
  \subfloat[Theoretical and experimental results of expected client buffer size $S$ for $k=1$ and $k=2$.]{\includegraphics[ width=4.4cm, height=4.4cm]{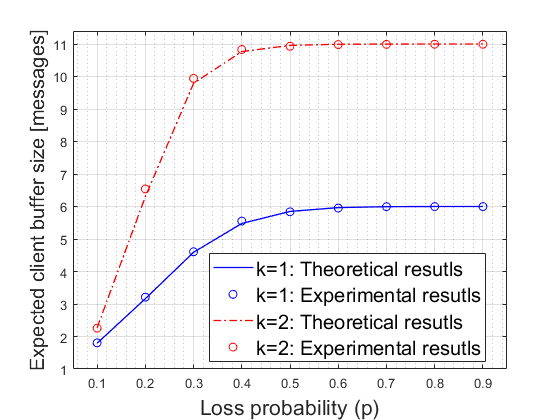}
  \label{Average_buffer_size_m_equals_5}}
  \subfloat[Theoretical results of probability of steady-states for $k=1$.]{\includegraphics[ width=4.4cm, height=4.4cm]{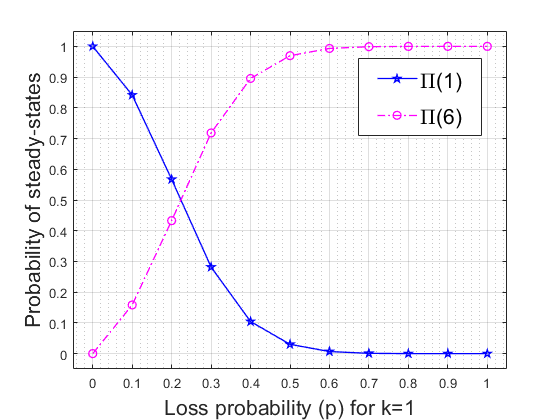}
  \label{probability_Of_steady_states_k_equals_1}}
  \subfloat[Theoretical results of probability of steady-states for $k=2$.]{\includegraphics[ width=4.4cm, height=4.4cm]{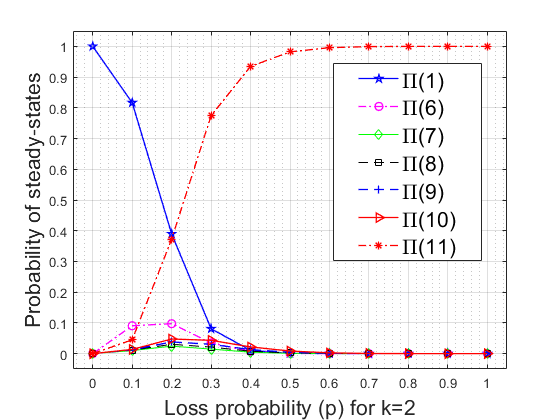}
  \label{probability_Of_steady_states_k_equals_2}}
  \subfloat[Experimental results of expected blocked arrival messages at an observation at the client for $k=1;2$.]{\includegraphics[ width=4.4cm, height=4.4cm]{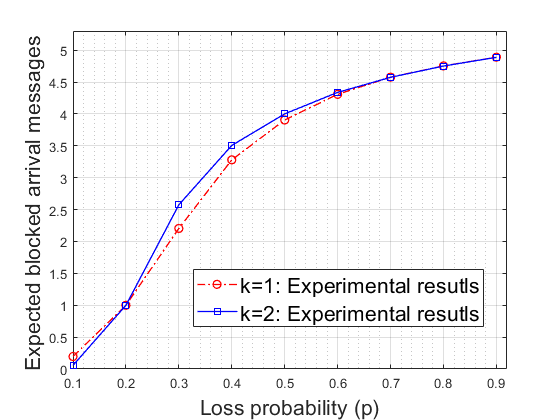}
  \label{messages_blocked}}
  \caption{Theoretical and experimental results for $k=1$ and $k=2$.}
  \label{Average_arrival_messages_blocked_TheoryAndPractical}
  \vspace{-0.5cm}
  \end{figure*}
The client uses python request library and protocol HTTP/$1.1$ to send POST messages to the server. A POST request message with the structure  $message = requests.post('http://192.168.1.2:5000/api/sensors', timeout=T_o, data=POST\_Data)$ is sent to the Raspberry Pi, where $192.168.1.2$ is Raspberry Pi's IP address, timeout event is $T_o=15$ seconds and $POST\_Data$ is the JSON file randomly created, as shown in the example of Fig. \ref{json}. The arrival time interval of each message updated at a time at the client is set $t=3s$. Hence, the number of arrival messages at the client in the timeout interval is $m=\frac{T_o}{t}=5$ messages. Each message has the header  $\{Header: ID\}$ identified for it, e.g., message $q_i$ has $ID=i$. The message length is $199$ bytes. The buffer length is $M=km+1$ messages, where for k=1, the buffer length is $M=6$ messages, and for k=2, we obtain $M=11$ messages. Each time period of 8.3 hours was observed to experimentally measure for each  result point in Fig. \ref{Average_buffer_size_m_equals_5} and Fig. \ref{messages_blocked}.
\subsection{Evaluation}
Fig. \ref{Average_buffer_size_m_equals_5} presents the theoretical and experimental results of expected buffer size $S$ for $k=1$ ($M=6$) and $k=2$ ($M=11$), vs. message loss probability $p$. Theorem \ref{theorem:ABS} and Theorem \ref{ABS:theoremk2} are used to draw the theoretical results for  $k=1$ and $k=2$. For $0 \leq p \leq 0.5$, the expected buffer size $S$ obtained by the theory and experiment increases with increasing $p$ because the buffer stores more messages, when $p$ increases. Nevertheless, when $p\geq 0.5$, $S$ approximately achieves the constant value of $M$ because the buffer is often full in this loss probability interval. The larger the buffer length $M$ is, the higher  the expected buffer size $S$ is because a higher value $M$ has a greater storage capacity. Based on corollary \ref{lexuan}, we have $\lim_{p \rightarrow 0} S=1$ for all $k$, $\lim_{p \rightarrow 1} S=6$ for $k=1$ and $\lim_{p \rightarrow 1} S=11$ for $k=2$. For $p < 0.5$, there is a small difference between the theoretical and experimental results, where the difference increases in terms of the value $k$, e.g., at $p=0.2$, for $k=1$, $S=3.164$ messages for theory and $S=3.211$ messages for experiment (experiment increases by $1.49 \%$) and for $k=2$, $S=6.302$ messages for theory and $S=6.540$ messages for experiment (experiment increases by $3.78 \%$). The reason of the insignificant difference is that: For theory, we assume that new messages arrived at the client are ignored when previous messages are stored in the buffer, and the client can send  successfully all the stored  messages with smooth connections. This can be explained by the fact that the mean arrival time $t=3$s is much higher than the transmission time of stored messages, e.g., in table \ref{timetime} for $10$ messages stored in the  buffer, we only consume $0.161$s to successfully send them with smooth connections, so  the ignored data  in our assumption is insignificant. For $p \geq 0.5$, the theoretical and experimental results are mostly overlapped because almost new messages arrived at the client are from timeout events, which the experimental model fits with the theoretical one. 
\begin{table}[]%
\caption{Examples of experimental transmission time for successfully sending  messages stored in the buffer with smooth connections.}
\label{timetime}
\begin{center}
\setlength\tabcolsep{3pt}
\begin{tabular}{|c|c|}
  \hline
Number of messages stored in the client buffer & Total time in seconds \\ \hline
$1$ & 0.034 \\ \hline
$3$ & 0.048 \\ \hline
$5$ & 0.089 \\ \hline
$8$ & 0.122 \\ \hline
$10$ &  0.161 \\ \hline
\end{tabular}
\end{center}
  \vspace{-0.6cm}
\end{table}

Fig. \ref{probability_Of_steady_states_k_equals_1} and Fig. \ref{probability_Of_steady_states_k_equals_2} present theoretical results of steady-states probability with buffer length of $k=1$ and $k=2$, respectively, vs. message loss probability $p$. Lemma \ref{kequal1} and Lemma \ref{lemma:21k2}  are used for Fig. \eqref{probability_Of_steady_states_k_equals_1} and Fig. \eqref{probability_Of_steady_states_k_equals_2}, respectively. For the steady-state probability $\pi(1)$ in Fig. \eqref{probability_Of_steady_states_k_equals_1} and Fig. \eqref{probability_Of_steady_states_k_equals_2}, the higher the loss probability $p$ is, the lower it is. For the steady-state probability $\pi(km+1)$, i.e., $\pi(6)$ and $\pi(11)$ in Fig. \eqref{probability_Of_steady_states_k_equals_1} and Fig. \eqref{probability_Of_steady_states_k_equals_2}, respectively, the higher the loss probability $p$ is, the higher it is. For $\pi(e)$ ($e \in [6,km]$) in Fig. \ref{probability_Of_steady_states_k_equals_2}, $\pi(e)$ gradually increases and then tends to gradually decrease to $0$. The reason for those behaviors is that the number of arrival messages stored in the client buffer increases when $p$ increases. 
 
Next, we analyse the steady-state probability $\pi(km+1)$ in Fig. \ref{probability_Of_steady_states_k_equals_1} and Fig. \ref{probability_Of_steady_states_k_equals_2} to understand the impact of the value $k$ on the expected blocked arrival messages $B$ at the client, i.e.,  $\pi(6)$ in Fig. \ref{probability_Of_steady_states_k_equals_1} and $\pi(11)$ in Fig. \ref{probability_Of_steady_states_k_equals_2}. For  $p \in [0.1;0.2]$, we obtain $\pi(6)>\pi(11)$, e.g., when $p=0.1$, we have $\pi(6)=  0.158$ and $\pi(11)=0.045$, this can be explained by the fact that with a higher  value $k$ we have a higher number of states, therefore we obtain a lower value $\pi(km+1)$. For $p \in  [0.6;1]$ we obtain the same value $\pi(km+1)$, e.g., $\pi(6)\approx\pi(11) \approx 1$, because with a very high loss probability $p$, the buffer is always full in this probability interval. In the interval of comparatively high value  $p \in [0.3;0.5]$,  we have $\pi(6) < \pi(11)$, e.g., $\pi(6)=0.718$ and $\pi(11)=0.774$ at $p=0.3$. In this case, the storage capacity of new arrival messages increases with the increase of the value $k$, i.e., increasing transmission chances of messages leads to increasing timeout events in this probability interval, which can increase the number of blocked arrival messages. In Fig. \ref{messages_blocked}, we show experimentally the expected blocked arrival messages $B$ at an observation at the client, whereby the observation is defined in section \ref{Analysis},  $\lim_{p \rightarrow 0} B = 0$, and $\lim_{p \rightarrow 1} B = m=5$ messages, i.e., $5$ messages arrived at the client are blocked in the timeout observation  when $p$ achieves $1$. With $p \in [0.3;0.5]$ in Fig. \ref{messages_blocked}, $B$ of $k=2$ is higher than $B$ of $k=1$, which confirms the theoretical analysis above. Through Fig. \ref{messages_blocked}, we recommend setting $k>1$ if $p<0.2$ and $k=1$ if $p\geq 0.2$ for reducing the blocked data. 

 Through the analysis above, we can derive an interesting conclusion that a large buffer length for reliable IoT systems in lossy environments does not always yield the best performance. Here, the amount of blocked data can increase, since a large $M$ value can cause more timeout events leading to increasing the number of arrival messages at the buffer, while the buffer is already filled with a large number of messages. This is consistent with discussions showed in \cite{inproceedings} for TCP.
 \section{Conclusion}\label{conclusion}
We benchmarked the buffer size in IoT devices deploying REST HTTP, in theory and experiment. The results showed that a large buffer in IoT devices does not always lead to an improved performance in lossy environments, and in fact could even degrade the performance. The proper benchmaking of the buffer size is hence rather important. The experimental analysis indicated that in order to reduce the blocked data, we should set $k>1$ if $p<0.2$ and $k=1$ if $p\geq 0.2$.
\bibliographystyle{IEEEtran}
\bibliography{nc-rest}

\begin{thebibliography}{10}
\providecommand{\url}[1]{#1}
\csname url@samestyle\endcsname
\providecommand{\newblock}{\relax}
\providecommand{\bibinfo}[2]{#2}
\providecommand{\BIBentrySTDinterwordspacing}{\spaceskip=0pt\relax}
\providecommand{\BIBentryALTinterwordstretchfactor}{4}
\providecommand{\BIBentryALTinterwordspacing}{\spaceskip=\fontdimen2\font plus
\BIBentryALTinterwordstretchfactor\fontdimen3\font minus
  \fontdimen4\font\relax}
\providecommand{\BIBforeignlanguage}[2]{{%
\expandafter\ifx\csname l@#1\endcsname\relax
\typeout{** WARNING: IEEEtran.bst: No hyphenation pattern has been}%
\typeout{** loaded for the language `#1'. Using the pattern for}%
\typeout{** the default language instead.}%
\else
\language=\csname l@#1\endcsname
\fi
#2}}
\providecommand{\BIBdecl}{\relax}
\BIBdecl

\bibitem{OASIS2014}
OASIS, ``{MQTT Version 3.1.1},'' \emph{OASIS Standard}, p.~81, 2014.

\bibitem{Shelby2014}
Z.~Shelby and C.~Hartke, K.~Bormann, ``{rfc7252, The Constrained Application
  Protocol (CoAP)},'' pp. 1--112, 2014.

\bibitem{Severance2015}
C.~Severance, ``{Roy T. Fielding: Understanding the REST Style.}''
  \emph{Computer}, vol.~48, no.~6, pp. 7--9, 2015.

\bibitem{10.1145/3292674}
\BIBentryALTinterwordspacing
J.~Dizdarevi\'{c}, F.~Carpio, A.~Jukan, and X.~Masip-Bruin, ``A survey of
  communication protocols for internet of things and related challenges of fog
  and cloud computing integration,'' \emph{ACM Comput. Surv.}, vol.~51, no.~6,
  Jan. 2019. [Online]. Available: \url{https://doi.org/10.1145/3292674}
\BIBentrySTDinterwordspacing

\bibitem{Heidemann:1997:MPH:268715.268719}
J.~Heidemann, K.~Obraczka, and J.~Touch, ``Modeling the performance of http
  over several transport protocols,'' \emph{IEEE/ACM Trans. Netw.}, vol.~5,
  no.~5, pp. 616--630, Oct. 1997.

\bibitem{HansKruse}
H.~Kruse, M.~Allman, J.~Griner, and D.~Tran, ``Experimentation and modelling of
  http over satellite channels,'' \emph{Inter. Jour. of Satellite
  Communications}, vol.~16, no.~1, pp. 51--68, 2001.

\bibitem{1258974}
P.~{Vaderna}, E.~{Stromberg}, and T.~{Elteto}, ``Modelling performance of
  http/1.1,'' in \emph{GLOBECOM '03}, vol.~7, Dec 2003, pp. 3969--3973 vol.7.

\bibitem{8400067}
W.~{Bziuk}, C.~V. {Phung}, J.~{Dizdarevic}, and A.~{Jukan}, ``On http
  performance in iot applications: An analysis of latency and throughput,'' in
  \emph{2018 MIPRO}, May 2018, pp. 0350--0355.

\bibitem{8756782}
C.~V. {Phung}, J.~{Dizdarevic}, F.~{Carpio}, and A.~{Jukan}, ``Enhancing rest
  http with random linear network coding in dynamic edge computing
  environments,'' in \emph{2019 MIPRO}, May 2019, pp. 435--440.

\bibitem{9149026}
C.~V. {Phung}, J.~{Dizdarevic}, and A.~{Jukan}, ``An experimental study of
  network coded rest http in dynamic iot systems,'' in \emph{ICC 2020}, pp.
  1--6.

\bibitem{7226719}
K.~{Shuang}, T.~{Zhang}, Z.~{Dong}, and P.~{Xu}, ``Impact of http pipelining
  mechanism for web browsing optimization,'' in \emph{2015 IEEE International
  Conference on Mobile Services}, June 2015, pp. 415--422.

\bibitem{inproceedings}
J.~Khan, M.~Shahzad, and A.~Butt, ``Sizing buffers of iot edge routers,'' 06
  2018, pp. 55--60.

\bibitem{6296043}
J.~{Edstrom} and E.~{Tilevich}, ``Reusable and extensible fault tolerance for
  restful applications,'' in \emph{2012 IEEE 11th International Conference on
  Trust, Security and Privacy in Computing and Communications}, 2012.

\bibitem{8088251}
N.~{Naik}, ``Choice of effective messaging protocols for iot systems: Mqtt,
  coap, amqp and http,'' in \emph{ISSE}, 2017, pp. 1--7.

\bibitem{7746084}
N.~A.~M. Alduais, J.~{Abdullah}, A.~{Jamil}, and L.~{Audah}, ``An efficient
  data collection and dissemination for iot based wsn,'' in \emph{IEMCON},
  2016.

\end{thebibliography}
\appendix
\subsection{Balance equation of state $m+1$ for the case of $k=2$} \label{appen1}
Based on section \ref{MCM} and using Eq. \eqref{1tov2},  Eq. \eqref{lasttov1} and Eq. \eqref{lasttov2}, the balance equation \eqref{pt2-2} of state $m+1$ for $k=2$ is:
  \begin{equation}\label{pt2-2}
\begin{split}
&\pi(m+1) \cdot \left [(1-p)^{m+1}+ p\sum_{i=m+2}^{2m+1}(1-p)^{2m+1-i} \right ]\\&= \pi(1) \cdot p + p\sum_{i=m+2}^{2m +1} \pi(i) \cdot (1-p)^{i-1}
 \end{split}
 \end{equation}
 \subsection{Balance equation of state $e \in  [m+2,2m-1]$ for $k=2$} \label{appen2}
 Based on section \ref{MCM} and using Eq. \eqref{lasttov1}, Eq. \eqref{lasttov2} and Eq. \eqref{lasttov3}, the balance equations \eqref{pt3-1} of state $e$ for $k=2$ is:
\begin{equation}\label{pt3-1}
\begin{split}
&\pi(e) \left [(1-p)^e+p \left (\sum_{i=m+1}^{e-1}(1-p)^{e-i+m} \right. \right. \\& \left. \left. +\sum_{i=e+1}^{2m}(1-p)^{e-i+m}+\sum_{i=0}^{e-m-1}(1-p)^i \right ) \right ]\\&=\sum_{i=m+1}^{e-1}\pi(i)\cdot (1-p)^{i-e+m} p +\sum_{i=e+1}^{2m+1}\pi(i) \cdot (1-p)^{i-e+m}  p
 \end{split}
 \end{equation}
  \subsection{Balance equation of state $2m$ for $k=2$} \label{appen3}
Based on section \ref{MCM} and using Eq. \eqref{lasttov1}, Eq. \eqref{lasttov2} and Eq. \eqref{lasttov3}, the balance equation \eqref{pt3-3} of state $2m$ for $k=2$ is:
 \begin{equation}\label{pt3-3}
\begin{split}
&\pi(2m) \left [(1-p)^{2m}+   p \left (\sum_{i=m+1}^{2m-1}(1-p)^{3m-i} +\sum_{i=0}^{m-1}(1-p)^i \right ) \right ]\\&=\left [\sum_{i=m+1}^{2m-1}\pi(i)\cdot (1-p)^{i-m}  p \right ] + \pi(2m+1) \cdot (1-p)^{m+1}  p
 \end{split}
 \end{equation}
  \subsection{Balance equation of state $2m+1$ for $k=2$} \label{appen4}
    Based on  section \ref{MCM} and using Eq. \eqref{lasttov1}, Eq. \eqref{lasttov2} and Eq. \eqref{lasttov3}, the balance equation \eqref{pt3-5} of state $2m+1$ for $k=2$ is:
\begin{equation}\label{pt3-5}
\begin{split}
&\pi(2m+1)\left [(1-p)^{2m+1}+p \sum_{i=m+1}^{2m}(1-p)^{3m-i+1} \right ]\\&= \sum_{i=m+1}^{2m} \pi(i) \cdot \sum_{j=0}^{i-m-1}(1-p)^j  p
 \end{split}
 \end{equation}  
\end{document}